\documentclass[12pt]{amsart}
\usepackage[margin=3cm]{geometry}
\usepackage{amssymb}
\usepackage{latexsym}
\usepackage[all,cmtip,arrow,2cell]{xy}
\usepackage{graphicx}
\usepackage{hyperref}
\usepackage{comment}
\usepackage{mathtools}
\mathtoolsset{showonlyrefs=true}

\theoremstyle{plain}
\newtheorem{thm}{Theorem}[section]
 
\newtheorem{lemma}[thm]{Lemma}

\theoremstyle{definition}

\newtheorem{ex}[thm]{Example}

 \newcommand{\half}{\textstyle{\frac{1}{2}}}
 \newcommand{\frakg}{\mathfrak{g}}

\newcommand{\arrows}{\,\lower1pt\hbox{$\longrightarrow$}\hskip-.24in\raise2pt
             \hbox{$\longrightarrow$}\,}

\begin{document}


\title{Plasma in monopole background is not twisted Poisson}
\thanks{C. Sard\'on acknowledges the hospitality  of the Department of Mathematics, University of California, Berkeley, during her stay as  Visiting Scholar.  The authors would like to thank Phillip Morrison and Thomas Strobl for useful comments on the manuscript.}%

----------
--------
\author{Manuel Lainz}
 \address{Instituto de Ciencias Matem\'aticas (CSIC-UAM-UC3M-UCM),
   calle Nicol\'as Cabrera, 13-15, Campus Cantoblanco, UAM
   28049 Madrid, Spain}
   \email{manuel.lainz@icmat.es}
   
\author{Cristina Sard\'on}%
\address{%
 Instituto de Ciencias Matem\'aticas (CSIC),
   calle Nicol\'as Cabrera, 13-15, Campus Cantoblanco, 
   28049 Madrid, Spain
}%
 \email{cristinasardon@icmat.es}

\author{Alan Weinstein}
\address{Department of Mathematics, University of California, Berkeley, CA 94720, USA}
\email{alanw@math.berkeley.edu}


\begin{abstract}
For a particle in the magnetic field of a cloud of monopoles, the naturally associated 2-form on phase space is not closed, and so the corresponding bracket operation on functions does not satisfy the Jacobi identity.  Thus, it is not a Poisson bracket; however, it is twisted Poisson in the sense that the Jacobiator comes from a closed 3-form.  

The space $\mathcal D$ of densities on phase space is the state space of a plasma.  The twisted Poisson bracket on phase-space functions gives rise to a bracket on functions on $\mathcal D$.
In the absence of monopoles, this is again a Poisson bracket.
It has recently been shown by Heninger and Morrison that this bracket is not Poisson when monopoles are present.  In this note, we give an example where it is not even twisted Poisson.
\end{abstract}

\maketitle
\section{Introduction}
\label{sec-intro}
The Vlasov-Maxwell equations describe the time-evolution in the continuum limit of a collisionless plasma of charged particles interacting via the electromagnetic fields which they generate.  It was shown by Morrison \cite{mo:,we-mo:} and Marsden and Weinstein \cite{MW1} that these equations form a hamiltonian system with respect to a Poisson structure derived by reduction from a symplectic structure on a product of cotangent bundles.

The fact that the Poisson structure satisfies the Jacobi identity is related to the vanishing of the divergence of the magnetic field, which depends on the absence of monopoles.  It has
been known for some time (see, for instance, \cite{mo2:}) that, when monopoles are present, the Jacobi identity is violated, so that the system is not hamiltonian, thus raising questions about the possibility of its quantization, and perhaps even about its physical validity.

Now the motion of a single charged particle in a divergence-free background magnetic field is well known to be hamiltonian.  When there is a global vector potential for the magnetic field, the Poisson bracket on position-momentum phase space may be taken to be the canonical one, with the vector potential appearing as part of the hamiltonian function generating the motion.  Alternatively, using the kinetic energy metric to convert the vector potential to a 1-form $\mathbf A$, one may eliminate its contribution to the hamiltonian by adding to the symplectic structure on phase space the pullback from configuration space of the magnetic field, represented as the closed 2-form $\mathbf B=d\mathbf A$.  In this way, one can express in hamiltonian form the equations of motion when there is no global vector potential, such as in the field outside a magnetic monopole.  (See, for instance, 
\cite{ma-ra:mechanics}.)  The original approach by 
Dirac \cite{di:} to this situation used a vector potential which was singular along a ``string" stretching from the monopole to infinity.

Wu and Yang \cite{wu-ya:} used the formalism of connections on vector bundles to avoid the singular strings.\footnote{In the quantum version, the wave functions take their values in a complex line bundle with a hermitian connection whose curvature is $\bf B$.  This approach is also useful for understanding the Aharonov-Bohm effect in the complement of a charge-carrying wire which is shielded in such a way that the magnetic field vanishes outside a neighborhood of the wire.  Here, the connection is flat but has nontrivial holonomy along paths encircling the wire; this affects the quantum but not the classical behavior of the particle.}

When the charged particle moves in a background magnetic field with nonzero divergence created by a ``cloud" of monopoles, the 2-form $\mathbf B$ is not even closed, and so the resulting ``Poisson bracket" on phase space no longer satisfies the Jacobi identity.  This fact came to our attention in Heninger and Morrison \cite{HM1}, but there are many previous references to this fact, as cited in \cite{HM1}.    However, the failure of the Jacobi identity is expressed in terms of a closed 3-form, making the bracket into what is known as a ``twisted Poisson bracket" \cite{al-st:current,ku-sz:symplectic, SW1}. This leads to the question of whether the plasma bracket, known not to be Poisson, is at least twisted Poisson.   

In this note, we consider a simplified system in which a plasma of charged particles, none of them monopoles, interacts with a background magnetic field which may not be divergence-free.  We show that the plasma bracket in this case can fail to be twisted Poisson, even though the single-particle bracket is twisted Poisson.
To show this, we use the fact \cite{SW1} that the image of the structural bivector for a twisted Poisson structure is an integrable (generally singular) distribution; thus, we will give an example of a magnetic field for which this distribution is not integrable.\footnote{In a paper on twisted Poisson structures in the dynamics of systems with nonholonomic constraints, Balseiro and Garc\'ia-Naranjo \cite{bana:} prove the converse result that, under a constant rank assumption, integrability of the image distribution of a bivector implies that it is a twisted Poisson structure for a suitable choice of 3-form.  Their paper also gives a nice general exposition of twisted Poisson geometry.}

\section{Plasma dynamics}
\subsection{Maxwell-Vlasov equations}
The usual Maxwell-Vlasov equations are a system of partial differential equations for the time evolution of quantities $f_s(\mathbf x,\mathbf v,t)$, $\mathbf E(\mathbf x,t)$, and $\mathbf B(\mathbf x,t)$.  Each $f_s$ is the density in phase space (position $\mathbf x$ and velocity $\mathbf v$) of a particle species $s$ with mass $m_s$ and electric charge $e_s$.  $\mathbf E$ and $\mathbf B$ are, as usual, the electric and magnetic fields.  The equations express the transport of phase space density for each species under the motion determined by the Lorentz equations, along with Maxwell's equations for the electromagnetic field, taking into account the charge density $\rho=\sum_s e_s\int{f_s d{\mathbf v}}$ and current density
$\mathbf j=\sum_s e_s\int{f_s {\mathbf v} d{\mathbf v}}$.

In the hamiltonian formulation of Morrison and Marsden/Weinstein, the hamiltonian functional is
\begin{equation}
\mathcal{H}=\sum_{s} \frac{m_s}{2}\int{|v|^2 f_s d{\mathbf x} d{\mathbf v}}+\frac{1}{8\pi}\int{\left(|{\mathbf E}|^2+|{\mathbf B}|^2\right)d{\mathbf x}},
\end{equation}
where $\mathbf E$ and $\mathbf B$ are determined by the $f_s$ according to the Maxwell equations which relate them to the charge and current densities.

When monopoles are admitted, we add to the particle data some nonzero monopole strengths $g_s$, leading to a monopole density and monopole current density which affect the definitions of $\mathbf E$ and $\mathbf B$, but the form of the hamiltonian remains the same.   

The full Poisson bracket, monopole terms included, is \cite{HM1}: 
\begin{align}\label{bracketmonopole}
\{F,G\}&=\sum_s \frac{1}{m_s}\int{f_s\left\{ F_{f_s}, G_{f_s}\right\}_{\mathrm{CAN}}  d{\mathbf x}d{\mathbf v}}
\\&+\sum_s \frac{e_s}{m_s^2c} \int{f_s{\mathbf B} \cdot \left(\frac{\partial F_{f_s}}{\partial {\mathbf v}} \times\frac{\partial G_{f_s}}{\partial {\mathbf v}}\right) d{\mathbf x}d{\mathbf v}}\nonumber
\\&-\sum_s \frac{g_s}{m_s^2c} \int{f_s{\mathbf E} \cdot \left(\frac{\partial F_{f_s}}{\partial {\mathbf v}} \times\frac{\partial G_{f_s}}{\partial {\mathbf v}}\right) d{\mathbf x}d{\mathbf v}}\nonumber \\
&+\sum_s\frac{4\pi e_s}{m_s} \int{f_s \left(G_{{\mathbf E}}\cdot\frac{\partial F_{f_s}}{\partial {\mathbf v}}
- F_{{\mathbf E}}\cdot\frac{\partial G_{f_s}}{\partial {\mathbf v}}\right)d{\mathbf x}d{\mathbf v}}\nonumber\\
&+\sum_s\frac{4\pi g_s}{m_s} \int{f_s \left(G_{{\mathbf B}}\cdot\frac{\partial F_{f_s}}{\partial {\mathbf v}}
- F_{{\mathbf B}}\cdot\frac{\partial G_{f_s}}{\partial {\mathbf v}}\right)d{\mathbf x}d{\mathbf v}}\nonumber\\
&+4\pi c\int{ \left(F_{{\mathbf E}}\cdot \left(\nabla \times G_{{\mathbf B}}\right)-G_{{\mathbf E}}\cdot\left(\nabla \times F_{{\mathbf B}}\right)\right) d{\mathbf x}}\nonumber.
\end{align}
Here, the subscripts on the functionals $F$ and $G$ of $f_s$, $\mathbf E$, and $\mathbf B$ denote the functional derivatives w.r.t.
the subscript variables, e.g.~$F_{f_s}=\delta F/\delta f_{s}$, and  $\left\{~, ~\right\}_{\mathrm{CAN}}$ is the usual canonical bracket on functions on $(\mathbf x,\mathbf v)$ phase space.

 The jacobiator of the bracket is \cite{HM1}:

{\begin{footnotesize}
\begin{align}\label{eq-bigjacobiator}
\{F,\{G,H\}\}+\text{cyc}=\sum_s \int{f_s\left(e_s\nabla \cdot {\mathbf
B} -g_s\nabla \cdot {\mathbf
E} \right)\left[\frac{\partial H_{f_s}}{\partial {\mathbf v}}\cdot
\left(\frac{\partial F_{f_s}}{\partial {\mathbf v}} \times\frac{\partial
G_{f_s}}{\partial {\mathbf v}}\right)\right]d{\mathbf x}d{\mathbf v}}.
\end{align}
\end{footnotesize}}

In the usual Maxwell-Vlasov situation, where there are no monopoles,  $\nabla \cdot \mathbf B $ and all the $g_s$ are zero, and so the jacobiator vanishes; i.e.~the Jacobi identity is satisfied.  In the general case, though, it is not.

\subsection{Particle dynamics}
\label{sec-nottwisted}
We next look at the dynamics of a single particle in an electromagnetic field.   Later, we will see how the plasma dynamical equations above are related to this.

For the motion of a particle with mass $m$ and charge $e$ in an electromagnetic field, the configuration space $Q$ is three-dimensional euclidean space, with an electric potential function $\phi$ and a magnetic vector potential $\mathbf A$.
The electric field is $\mathbf E=\nabla \phi$ and the magnetic field is $\mathbf B=\nabla \times \mathbf A$.  In Newton's equation of motion $\mathbf F = m \mathbf a$,  $\mathbf F$ is the Lorentz force
$e(\mathbf E + \mathbf v \times \mathbf B)$, a function of both  position and velocity.  

For our purposes, and to allow for generalization, we will use the hamiltonian formalism in which $\mathbf A$ becomes a 1-form via the identification between tangent and cotangent vectors induced by the riemannian metric on euclidean space,
and $\mathbf B$ is the 2-form $d\mathbf A$.  (For more details, see, for example
\cite{ma-ra:mechanics}.)  We can then use the following general setup.
The configuration space $Q$ can be any manifold (not necessarily of dimension 3).  The magnetic field is a 2-form $\mathbf B$ on $Q$.  The Maxwell equation $\nabla \cdot \mathbf B = 0$
 becomes the condition $d\mathbf B = 0$.  We define a new phase space $T^* _\mathbf B Q$ to be $T^*Q$ with the symplectic form $\omega_\mathbf B$ which is the sum of the canonical form $\omega_Q$ and the pullback of $\mathbf B$ by the natural projection $T^*Q\to Q$.  (For convenience, we will assume the charge $e$ to be unity.   If $\mathbf A$ is any 1-form on $Q$, the ``gauge transformation" given by translation by $\mathbf A$ takes $\omega_\mathbf B$ to $\omega_{\mathbf B -d\mathbf A}.$  In particular, $T^*_\mathbf B Q$ is symplectomorphic to $T^*Q$ if and only if $\mathbf B$ is exact.  We assume that $Q$ is provided with a ``kinetic energy" riemannian metric, which implicitly includes the mass, so that the metric identifies velocity $\mathbf v$ in $TQ$ with momentum in $T^*Q$; by abuse of notation, we will also denote the momentum itself by $\mathbf v$.   The hamiltonian for the Lorentz flow is then $\half ||\mathbf v||^2 + \phi$, with the magnetic field encoded in the symplectic form $\omega_\mathbf B$. 

For motion in the magnetic field created by a smooth distribution of monopoles, we must drop the assumption that $d\mathbf B = 0$.
Our goal is to see what remains of the hamiltonian theory, now that the 2-form $\omega_\mathbf B$ is no longer closed.  It is still nondegenerate, though, and so we can ``invert" it to produce a bivector $\pi_\mathbf B$.  As we will see in Section \ref{subsec-twisted}, 
$\pi_\mathbf B$, though no longer necessarily a Poisson structure, is twisted Poisson  for the closed 3-form $\phi = d\mathbf B$.

\section{Twisted Poisson structures}
\label{subsec-twisted}
Twisted Poisson structures were  introduced by Klim\v{c}ik and Strobl \cite{kl-st:} (who called them WZW Poisson structures), inspired by previous work in string theory with closed 3-form ``backgrounds".  They were then named and studied by \v{S}evera and Weinstein  \cite{SW1}  in the setting of Dirac structures.

Just as Poisson structures on a manifold $M$
may be identified with certain Dirac structures in the standard
Courant algebroid $E_0=TM\oplus T^*M$; see Courant \cite{co:dirac}
and Liu {\it et al.}~\cite{li-we-xu:manin}, twisted Poisson structures may be defined as Dirac structures in a modified Courant algebroid $E_\phi$ in which the term $\phi(X_1,X_2,\cdot)$ is added to the  right hand side of the definition
\begin{equation}
\label{eq-courant}
[(X_1, \xi_1),( X_2,\xi_2)]=
([X_1,X_2], {\mathcal L}_{X_1}\xi_2-i_{X_2}d\xi_1),
\end{equation}
of the standard bracket on  $E_0$.

The graph of $\tilde{\pi}$ for a  bivector field\footnote{Here (and analogously for other bilinear forms), $\tilde{\pi}$ is the bundle map from $T^*M$ to $TM$ defined by $\alpha(\tilde{\pi}(\beta))=\pi(\alpha,\beta)$ for 1-forms $\alpha$ and $\beta$.}
$\pi$ on $M$ turns out to be a Dirac structure in $E_\phi$ if
and only if it satisfies the equation
\begin{equation}
\label{eq-quasi}
[\pi,\pi]=2\wedge^3\tilde{\pi}(\phi).
\end{equation}
We call such bivectors, and the associated brackets, {\bf twisted Poisson structures}.  We will also refer to the associated linear map $\tilde{\pi}$ as a twisted Poisson structure.

On the other hand, the graph of $\tilde{\omega}$ for a 2-form $\omega$ on $M$ is a Dirac structure in $E_{\phi}$ if and only if 
$d\omega=\phi$.  In particular, if $\pi$ is a nondegenerate bivector, so that it comes from a 2-form $\omega$ with the same graph, $\pi$ is necessarily a twisted Poisson structure with $\phi=d\omega$.  

Poisson brackets and hamiltonian vector fields are defined as usual:
$\{f,g\}=\pi(df,dg)$ and $H_f=\{\cdot,f\}.$  
With these definitions, the Jacobi
identity becomes:
\begin{equation}
\label{eq-jacobi}
\{\{f,g\},h\} + \text{cyc}=- \phi(H_f,H_g,H_h).
\end{equation}

As with any Dirac structure, (the graph of) a twisted Poisson structure is a Lie algebroid, which here may be identified with the cotangent bundle.  The anchor of this Lie algebroid is the bundle map $\tilde{\pi}$, whose image is the set of values of all hamiltonian vector fields.  As is true for any Lie algebroid, the (generally singular) distribution consisting of these ``hamiltonian vectors" is integrable; the leaves carry non-degenerate (but generally not closed) 2-forms.

Finally, we note that 
\begin{equation}
\label{eq-brackets}
H_{\{f,g\}}+[H_f,H_g]=-\tilde\pi(\phi(H_f,H_g,\cdot)).\end{equation}
 In Section \ref{sec-nontwisted} below, we will use these facts to prove that a certain bivector field is NOT a twisted Poisson structure.

\section{Almost Poisson structure on the dual of an almost Lie algebra}
\label{sec-nontwisted}
Let $\mathfrak g$ be a finite-dimensional  Lie algebra.  It is well known that there is an induced Poisson structure on the dual space $\mathfrak g^*$, known as the {\bf Lie-Poisson structure}.  The bracket on functions is given by
\begin{equation}
\label{eq-liepoisson}
\{f_1,f_2\}(c) = c([df_1(c),df_2(c)]),
\end{equation} 
 where the values of $df_i(c)$ at $c\in \frakg^*$, which belong to $\frakg^{**},$ are considered as elements of $\frakg.$  An infinite-dimensional version of this comes from taking $\mathfrak g$ to be the space $\mathcal F(M)$ of compactly supported smooth functions on a Poisson manifold $M$, with the Poisson bracket Lie algebra structure, and $\mathfrak g^*$  the {\em topological} dual space $\mathcal D(M)$ of distributional densities on $M$.  (The identification of $\frakg^{**}$ with $\frakg$ remains valid in this situation.)

The proof that the bracket on (functions on) $\frakg^*$ satisfies the Jacobi identity depends, of course, on the Jacobi identity in the Lie algebra $\frakg$; it is natural to ask, then, if $\frakg = \mathcal F (M)$ where $M$ is a  twisted Poisson manifold, whether this property carries over to $\mathcal{D}(M)$ as well.  

An {\bf almost Lie algebra} is a vector space $\frakg$ with a bilinear operation which is antisymmetric but which does not necessarily satisfy the Jacobi identity.  Just as for a Lie algebra, there is a bivector on $\frakg^*$ defining a bracket operation by \eqref{eq-liepoisson}.   This bracket is a Poisson bracket, i.e.~it satisfies the Jacobi identity, if and only if $\frakg$ is actually a Lie algebra.

We will show below that, for a certain almost Lie algebra $\frakg$, the almost Poisson structure on $\frakg^*$ is not even twisted Poisson.  Following our observation in Section \ref{subsec-twisted}, it suffices to show that the distribution consisting of values of hamiltonian vector fields is not integrable.  We will do this in the following way.

To determine whether the distribution consisting of hamiltonian vector values on the dual $\mathfrak g^*$ of an almost Lie algebra is integrable, we must test whether the value at each point $f$ of the bracket
$[H_a,H_b]$ of the hamiltonian vector fields of any two functions $a$ and $b$ on $\mathfrak g^*$ is again the value of a hamiltonian vector field.  To show that this is NOT the case, it is sufficient to find {\em linear} functions for which the bracket is not the value at $f$ of the hamiltonian vector field of any linear function $h$, since linear functions have all possible differentials at each point.  Each such linear function $a$ corresponds to an element of  $\mathfrak g$.
In our situation, where $\mathfrak g$ is $\mathcal F(M)$, this element is a function on $M$ which we will also denote by $a$.

For an element $a$ in any almost Lie algebra $\mathfrak g$, the hamiltonian vector field on $\mathfrak{g}^*$ of the linear function corresponding to $a$ is the linear operator $H_a$ given by the ``coadjoint action" of $a$, by which we mean the operator dual to the ``adjoint" operator $A_a:b\mapsto [a,b].$  When $\mathfrak g$ is $\mathcal{F}(M)$ for an almost-Poisson manifold $M$, the bracket $[~,~]$
is the almost Poisson bracket, so we will denote it by $\{~,~\}$ instead.  The adjoint operator is then $b\mapsto \{a,b\}$, which is 
the negative of the hamiltonian vector field $H_{a}$ of the function $a$ on $M$, operating on functions.  The dual of this operator, acting on densities in $\mathcal D(M)$, is the negative of the Lie derivative operator $L_{H_{a}}.$

In the case which will be of interest below, the almost Poisson structure comes from a symplectic structure whose Liouville volume form is invariant under all hamiltonian vector fields.   This enables us to identify the densities with (generalized) functions on $M$, in which case the Lie derivative operator $L_{H_{a}}$ becomes identified with $H_{a}$ itself.  

So we are left with the problem on $M$ of finding functions $a$ and $b$, and a function $f$ (representing a density by multiplication with the Liouville volume form) for which 
$[H_{a},H_{b}]f$ is not equal to $H_{h}f$ for any function $h$.  By Equation \eqref{eq-brackets}, we have
$$[H_{a},H_{b}]f = -H_{\{a,b\}}f - \tilde{\pi}(\phi(H_{a},H_{b},\cdot))f.$$
Since the term $-H_{\{a,b\}}f$ IS the value at $f$ of a hamiltonian vector field, we are left with trying to show that 
$\tilde{\pi}(\phi(H_{a},H_{b},\cdot))f$ is NOT such a value.

On the other hand, for the value at $f$ of the hamiltonian vector field of a function $h$, we have $H_{h}f = \{f,h\}= -H_{f}h.$   Thus, we must show that 
$\tilde{\pi}(\phi(H_{a},H_{b},\cdot))f$ is not in the image of the hamiltonian vector field operator $H_{f}$.  We will use the following simple lemma.

\begin{lemma}
If a function $g$ on $M$ is in the range of the operator given by a vector field $\xi$, then the integral of $g$ around any closed orbit of $\xi$ must be zero.
\end{lemma}
\begin{proof}
If $\xi h=g$, then $g$ is the derivative of $h$ along each orbit of $\xi$.
If the orbit is closed, the value of $h$ repeats after a period of the orbit, so the integral of its derivative must be zero.
\end{proof}

It remains, then, to exhibit an example where a twisted Poisson structure $\pi$ on $M$ is the inverse of a nondegenerate 2-form whose associated volume form is invariant under every hamiltonian vector field, but where the integral of 
$\tilde{\pi}(\phi(H_{a},H_{b},\cdot))f$ around a closed orbit of the
hamiltonian vector field $H_{f}$ is not zero for some functions $a$, $b$,
and $f$.  This is what we shall now do, where $M$ is the phase space of a
particle in a magnetic field for a monopole background.

Now,

\begin{ex}\normalfont
\label{ex:nottwisted}
Let $M=\mathbb R^6 = T^* \mathbb R^3$ with coordinates
 $\{x_1,x_2,x_3,p_1,p_2,p_3\}$.  Let $\mathbf B = x_2^2 ~dx_2 \wedge dx_3+ x_1 x_2 dx_1 \wedge dx_3$, so that the symplectic form on $M$ is 
 $$\omega_{\mathbf B}=\sum_i dx_i \wedge dp_1 + x_2^2 ~dx_2 \wedge dx_3+ x_1 x_2~ dx_1 \wedge dx_3.$$  The twisted Poisson structure on $M$ is then
 $$\pi_{\mathbf B}= \sum_i \frac{\partial}{\partial x_i} \wedge 
 \frac{\partial}{\partial p_i} +
 x_2^2 ~\frac{\partial}{\partial p_2} \wedge \frac{\partial}{\partial p_3}+ x_1 x_2~ \frac{\partial}{\partial p_1} \wedge \frac{\partial}{\partial p_3}.$$

The almost Poisson structure on $\mathcal D(M)$ is the following simplified version of the Maxwell-Vlasov bracket \eqref{bracketmonopole}:
 \begin{equation}\label{bracketsimplified}
\{F,G\}=\int{f\left\{ F_{f}, G_{f}\right\}_{\mathrm{CAN}}  d{\mathbf x}d{\mathbf v}}
+ \int{f{\mathbf B} \cdot \left(\frac{\partial F_{f}}{\partial {\mathbf v}} \times\frac{\partial G_{f}}{\partial {\mathbf v}}\right) d{\mathbf x}d{\mathbf v}}\nonumber
\end{equation}

The jacobiator formula 
\eqref{eq-bigjacobiator}
then simplifies to 
\begin{equation*}
\{F,\{G,H\}\}+\text{cyc}=\int{f\left(\nabla \cdot {\mathbf
B} \right)\left[\frac{\partial H_{f}}{\partial {\mathbf v}}\cdot
\left(\frac{\partial F_{f}}{\partial {\mathbf v}} \times\frac{\partial
G_{f}}{\partial {\mathbf v}}\right)\right]d{\mathbf x}d{\mathbf v}}.
\end{equation*}
 
 We have the Schouten bracket:
\begin{equation}
   [\pi_{\mathbf B},\pi_{\mathbf B}]_{SN}=2x_1\frac{\partial }{\partial p_1}\wedge \frac{\partial }{\partial p_2}\wedge \frac{\partial }{\partial p_3},
    \end{equation}
 and $$\phi=d\omega_{\mathbf B}=
 -x_1 dx_1\wedge dx_2 \wedge dx_3.$$
 (Note that Equation \eqref{eq-quasi} is indeed satisfied.)

Let $f$ be the function $x_1p_2-x_2p_1.$  Its hamiltonian vector field is
\begin{equation}
H_{f}=x_1\frac{\partial}{\partial x_2}-x_2\frac{\partial}{\partial x_1}+p_1\frac{\partial}{\partial p_2}-p_2\frac{\partial}{\partial p_1}.
\end{equation}
(The ``$\mathbf B$" part of $\omega_{\mathbf B}$ contributes two terms in $\frac{\partial}{\partial p_3},$ but they cancel one another.)    On each  orbit of $H_{f}$, $x_3$ and $p_3$ are constant, while the projections in the $(x_1,x_2)$ and $(p_1,p_2)$ planes each traverse at unit speed circles centered at the origin.  In particular, all of the orbits are periodic.

 It is easy to check that the Liouville volume is preserved under the Hamiltonian flow, i.e., $\mathcal{L}_{H_{f}}\omega_{{\mathbf B}}^3=0$.  We have 
$ \mathcal{L}_{H_{f}}\omega_{\mathbf B}^3=3\omega_{{\mathbf B}}^2\wedge {\mathcal{L}}_{H_{f}}\omega_{\mathbf B}.$
But
 \begin{equation}
\mathcal{L}_{H_{f}}\omega_{{\mathbf B}}=x_1^2 dx_1\wedge dx_3+ x_1 x_2 dx_2\wedge dx_3,
 \end{equation}
and so,
\begin{equation*}
\omega_{{\mathbf B}}^2\wedge \mathcal{L}_{H_{f}}\omega_{{\mathbf B}}=0.
\end{equation*}

Now we choose $a=p_3,b=p_1$.  
For their hamiltonian vector fields, we have
$$H_{a}=H_{p_3}=\frac{\partial}{\partial x_3}  +
 x_2^2 ~\frac{\partial}{\partial p_2} + x_1 x_2~ \frac{\partial}{\partial p_1},
$$ and
$$  H_{b}=H_{p_1}=\frac{\partial}{\partial x_1} - x_1 x_2~ \frac{\partial}{\partial p_3}. $$
Contracting these with $\phi= -x_1 dx_1\wedge dx_2 \wedge dx_3$
gives $x_1 dx_2,$ applying $\tilde{\pi}$ gives $x_1 \frac{\partial}{\partial p_2},$ and operating on $f = x_1p_2-x_2p_1$ gives $x_1^2.$
The integral of this nonnegative function around almost every one of the periodic orbits of $H_{f}$ is obviously positive (in fact, it is just the number $\pi$). 

This shows that the distribution on $\mathcal{D}$ spanned by hamiltonian vector fields is not integrable; hence, our almost Poisson structure  is not twisted Poisson.

\end{ex}

\section{Discussion}
One may wonder why a twisted Poisson structure on phase space could fail to be twisted Poisson when lifted to the space of densities.  Here are some thoughts on a possible answer.  Any almost Lie bracket $\beta$ on a vector space $\frakg$ (in our case, the functions on phase space) ``lifts" to an almost Poisson structure on $\frakg^*.$  When $\beta$ satisfies the Jacobi identity, so does the lift. But the condition that $\beta$ be a twisted Poisson bracket involves more than the vector space structure on $\frakg$.  $\frakg$ must also have a multiplicative structure so that there is a notion of 3-form.  But we see no way of lifting this structure to the functions on $\frakg^*$ so as to construct a 3-form on $\frakg^*$ which would make the lifted bracket twisted Poisson.

It still might be the case that some almost Lie but not Lie brackets on $\frakg$
could lift to twisted Poisson structures on $\frakg^*$.  A particular case of interest would be that where $\frakg$ is the space of functions on the phase space of a particle in the magnetic field of a {\em uniform} distribution of monopoles.

Finally, we may ask whether, for the Maxwell-Vlasov bracket, there is any useful remnant of the fact that the single-particle bracket is twisted Poisson.


\begin{thebibliography}{10}
\bibitem{mo:}
P.J. Morrison,
Maxwell-Vlasov equations as a continuous Hamiltonian system,
{\em Physics Letters A} {\bf 80} (1980), 383--386.

\bibitem{we-mo:}
A. Weinstein, P.J. Morrison,
Comment on the Maxwell-Vlasov equations as a continuous Hamiltonian system, 
{\em Phys. Lett. A} {\bf 86} (1981), 235-236 . 

\bibitem{MW1}
J.E. Marsden and A. Weinstein,
The Hamiltonian structure of the Maxwell-Vlasov equations,
{\em Physica D: Nonlinear Phenomena} {\bf 4} (1982), 394-406 .

\bibitem{mo2:}
 P. J. Morrison, 
Poisson brackets for fluids and plasmas,
{\em AIP Conf. Proc.}
{\bf 88} (1982), 13-46.

\bibitem{ma-ra:mechanics}
J.E. Marsden and T. Ratiu,
{\em Introduction to Mechanics and Symmetry},
Texts in Applied Mathematics {\bf 17}, Springer,  2nd Ed. 2002.

\bibitem{di:}
P.A.M. Dirac, 
Quantised singularities in the electromagnetic field,
{\em Proc. Roy. Soc. A} {\bf 133} (1931), 60-72.

\bibitem{wu-ya:}
T.T. Wu and C.N. Yang,
Dirac's monopole without strings: Classical Lagrangian theory
{\em Phys. Rev. D} {\bf 14} (1976), 437-445.

\bibitem{HM1}
M. Heninger and P. Morrison, Hamiltonian nature of monopole dynamics (preprint), arXiv 1808.08689 (2018).

\bibitem{al-st:current}
Alekseev, A. and Strobl, T., Current algebras and differential geometry, {\em J. High Energy Phys.}  no. 3, (2005), 035, 14 pp.

\bibitem{ku-sz:symplectic}
V.G. Kupriyanov and R.J. Szabo, 
Symplectic realization of electric charge in fields of monopole distributions,
{\em Phys. Rev. D} {\bf 98} (2018), no. 4, 045005, 25 pp. 

\bibitem{SW1}
P. \v{S}evera and A. Weinstein
Poisson Geometry with a 3-Form Background,
{\em Progress of Theoretical Physics Supplement} {\bf 144} (2001), 145-154.

\bibitem{bana:}
P. Balseiro and L.C. Garc\'ia-Naranjo,
Gauge transformations, Twisted Poisson brackets and hamiltonization of nonholonomic Systems,
{\em Archive for Rational Mechanics and Analysis} {\bf 205} (2012), 267-310.

 




\bibitem{kl-st:}
C. Klim\v{c}ik and T. Strobl,
WZW-Poisson manifolds,
{\em J. Geom. Phys.} {\bf 4} (2002), 341-344.

\bibitem{co:dirac}
T.J. Courant, Dirac Manifolds,
{\em Trans. Amer. Math. Soc.} {\bf 319} (1990), 631-661.






\bibitem{li-we-xu:manin}
Z.J. Liu, A. Weinstein, and P. Xu,  Manin triples for Lie
bialgebroids, {\em J. Diff. Geom.} {\bf 45} (1997), 547-574.


\end{thebibliography}
\end{document}